\documentclass[11pt]{article}

\usepackage[a4paper,margin=30mm]{geometry}
\usepackage{amsmath,amssymb,amsthm,mathtools,amscd}
\usepackage{bm}
\usepackage{booktabs}
\usepackage{enumitem}
\usepackage{microtype}
\usepackage{xcolor}
\usepackage{jheppub}
\usepackage[nameinlink,capitalise,noabbrev]{cleveref}
\usepackage{array}
\usepackage{longtable}
\usepackage{mathrsfs}
\usepackage{setspace}

\hypersetup{
  colorlinks=true,
  linkcolor=blue!45!black,
  citecolor=blue!45!black,
  urlcolor=blue!55!black,
  pdftitle={Classical Shadows of Higher-Form BRST Anomalies from a Weil-Covariant Phase-Space Bicomplex}
}

\setlist{itemsep=0.25em,topsep=0.4em}

\newtheorem{theorem}{Theorem}[section]
\newtheorem{proposition}[theorem]{Proposition}

\theoremstyle{definition}
\newtheorem{definition}[theorem]{Definition}

\theoremstyle{remark}

\newcommand{\dH}{\mathrm d_H}
\newcommand{\dd}{\mathrm d}
\newcommand{\del}{\boldsymbol\delta}
\newcommand{\cP}{\mathcal P}

\newcommand{\cW}{\mathcal W}

\newcommand{\g}{\mathfrak g}
\newcommand{\CE}{\mathrm{CE}}

\newcommand{\Pol}{\operatorname{Pol}}
\newcommand{\degh}{\operatorname{degh}}

\newcommand{\Lie}{\mathcal L}
\newcommand{\ii}{\iota}
\newcommand{\R}{\mathbb R}

\title{\textbf{Classical Shadows of Higher-Form BRST Anomalies}}
\author[a,b]{Ruizhi Shen}
\author[a,c]{Fu-Wen Shu}
\affiliation[a]{Center for Relativistic Astrophysics and High Energy Physics, Nanchang University, Nanchang 330031, China}
\affiliation[b]{Jiluan academy, Nangchang University, Nanchang, 330031, China}
\affiliation[c]{Department of Physics, Nanchang University, Nanchang, 330031, China}
\emailAdd{rz.shern@gmail.com,shufuwen@ncu.edu.cn}

\abstract{We formulate a precise sense in which a BRST anomaly may possess a classical phase space shadow.   It is the non-equivariance class of a Hamiltonian symmetry action on covariant phase space.  To organize this statement, we introduce the Weil covariant phase space bicomplex,
whose basic subcomplex controls equivariant Hamiltonian lifts.  
We then compare this purely classical class with local BRST descent.  For a Hamiltonian admissible mixed transgression, the ghost-number-two descendant \(a^{2}_{d-1}\) defines a deghostification map to the charge algebra, which is generated by the failure of the Cartan representative to be equivariant. 
The induced cohomology class is independent of Bardeen type changes of descent representative.  Thus the classical anomaly is a shadow of the BRST anomaly in the precise sense that both are realizations of the same Weil transgression class, while remaining objects of different theories. As
an explicit example, we analyze a five dimensional inflow transgression and show that the classical charge cocycle and the BRST descendant arise as two realizations of the same Weil transgression class.}

\begin{document}
\maketitle

\newpage 

\section{Introduction}

A quantum anomaly is conventionally characterized by the failure of a classical symmetry to survive quantization.  In local BRST language, a consistent anomaly in \(d\) dimensions is represented by a relative cohomology class
\[
  [a_d^1]\in H^{1,d}(s\mid\dH),
  \qquad sa_d^1+\dH a_{d-1}^2=0,
\]
and it is realized quantum mechanically when an effective action satisfies
\[
  s\Gamma=\hbar\int_X a_d^1+O(\hbar^2).
\]
The cohomology class and its geometric representatives may be constructed using classical background fields, Chern--Weil theory and BRST ghosts, but the physical statement that a global symmetry cannot be consistently gauged is a statement about the quantum theory\cite{Gaiotto2015}.

The purpose of this paper is to isolate a different, entirely classical object and to prove that, for a controlled class of transgressive higher-form anomalies, it is the phase space shadow of the BRST anomaly.  The classical object is the failure of a Hamiltonian symmetry action to be equivariant.  If a Lie algebra \(\g\) acts on a symplectic or presymplectic phase space \((\cP,\Omega)\) with Hamiltonians \(H_\epsilon\), then
\[
  K(\epsilon_1,\epsilon_2)
  :=\{H_{\epsilon_1},H_{\epsilon_2}\}-H_{[\epsilon_1,\epsilon_2]}
\]
is constant on phase space and is a Chevalley--Eilenberg two-cocycle. A nonzero class \([K]\in H^2_{\CE}(\g;\R)\) means that the charge algebra realizes a central extension rather than the original symmetry algebra, which has been called a classical anomaly \cite{Toppan2001}. In field theory with boundaries, the relevant algebra is the algebra of nontrivial boundary gauge transformations and the cocycle is a corner integral.

The central conceptual step is to place this construction in a Weil model. The Weil covariant phase space bicomplex
\[
  W(\g_\partial)\otimes\Omega(\cP_\Sigma)
\]
separates two logically distinct issues. The Hamiltonian equation says that the symplectic form has a Cartan extension, while equivariance says that this extension is basic. The cocycle \(K\) is precisely the defect of basicness. This is a purely classical statement.

The BRST side is related but not identical. BRST ghosts provide a Chevalley-Eilenberg realization of the infinitesimal symmetry algebra, and the second descendant \(a_{d-1}^2\) is the degree two symmetry cocycle associated with the first descendant \(a_d^1\).  Under a Hamiltonian admissibility hypothesis, we prove that the deghostified, Cauchy-transgressed second descendant is exactly the classical charge cocycle. The resulting map
\[
  \mathfrak S_{\Sigma}:
  H^{1,d}_{\mathrm{adm}}(s\mid\dH)
  \longrightarrow H^2_{\CE}(\g_\partial;\R)
\]
will be called the \emph{classical shadow map}. It does not assert that a quantum anomaly is literally classical.  It asserts that the quantum BRST class has a canonical classical descendant in phase space cohomology whenever the first descent is Hamiltonian integrable.

Local BRST cohomology and anomaly descent are classical subjects \cite{BarnichBrandtHenneaux2000}; homotopy moment maps and their obstruction theory are well developed \cite{CalliesEtAl2016,FregierEtAl2015,Dinamo2026}; BRST covariant phase space has already been related to corner charges and holographic Ward identities, and more recently to BRST Noether corner charges, charge brackets and associated BRST cocycles, including rank one theories with 2-form gauge fields \cite{BaulieuWetzstein2024,Baulieu2026}; and central extensions of electric and magnetic boundary charges in Maxwell theory are known \cite{FreidelPranzetti2018,HofmanIqbal2018}. Our contribution is the explicit Weil--CPS comparison theorem that identifies the classical non-equivariance class with the second higher form BRST descendant, together with a representative level computation for the mixed one-form anomaly developed in \cite{JiaWangZhang2026I}.

The paper is organized as follows.  \Cref{sec:classical} defines the classical phase space anomaly and interprets it in the Weil bicomplex.  \Cref{sec:shadow} introduces the local BRST descent, the deghostification map and the shadow theorem.  \Cref{sec:example} gives the complete mixed two-form calculation, including the improved descent representative. \Cref{sec:discussion} concludes this paper, explains the distinction between classical and quantum anomalies and give an insight into the shadow map anology with the non-squeezing theorem.  Technical material is collected in the appendices.

\section{Classical phase space anomalies and the Weil--CPS bicomplex}
\label{sec:classical}
This section formulates the classical charge cocycle as the CPS anomaly in general classical phase space. Then we introduce Weil-CPS bicomplex to better interpret the anomaly and discuss why this bicomplex is important in realizing our classical anomaly. For the signs we use, see \cref{app:signs} 
\subsection{Covariant phase space with a corner}

Let \(Y\) be an oriented d+1 dimensional manifold, \(\Sigma\subset Y\) an oriented hypersurface and \(S=\partial\Sigma\) its corner. Let \(L\) be a local Lagrangian \((d+1)\)-form on the infinite jet bundle of a classical field bundle. We use the first variation decomposition
\begin{equation}
  \del L=\mathcal E+\dH\theta,
  \qquad \omega:=\del\theta,
  \label{eq:firstvariation}
\end{equation}
where \(\mathcal E\) is the Euler--Lagrange form, \(\theta\) the presymplectic potential current and \(\omega\) the presymplectic current.  On shell and on linearized solutions, \(\dH\omega\simeq0\). The integrated presymplectic form is
\begin{equation}
  \Omega_\Sigma:=\int_\Sigma\omega.
  \label{eq:OmegaSigma}
\end{equation}
As usual, one either reduces by the degeneracy directions that vanish at \(S\), or retains the transformations with nontrivial boundary values as physical boundary symmetries. We denote the resulting classical phase space by \(\cP_\Sigma\) and its boundary symmetry algebra by \(\g_\partial\). We handle the boundary terms and integrability condition here as the standard ones with the Hamiltonian generators exist \cite{HarlowWu2019,CattaneoMnevReshetikhin2012}.

We adopt the Hamiltonian convention
\begin{equation}
  \ii_{X_F}\Omega_\Sigma=-\del F,
  \qquad \{F,G\}:=\Omega_\Sigma(X_F,X_G).
  \label{eq:HamConvention}
\end{equation}
If \(\epsilon\in\g_\partial\) generates a vector field \(v_\epsilon\) on \(\cP_\Sigma\), a Hamiltonian lift is a linear assignment \(\epsilon\mapsto H_\epsilon\) such that
\begin{equation}
  \ii_{v_\epsilon}\Omega_\Sigma=-\del H_\epsilon.
  \label{eq:HamLift}
\end{equation}

\subsection{The classical charge cocycle}
We then define the non-equivariance cocycle as our classical anomaly, because non-equivariance means that the poisson algebra structure fails in interpreting the Lie algebra structure, and in the ideal case, one would expect \(\epsilon\mapsto H_{\epsilon}\) to be a Lie algebra homomorphism.

\begin{definition}[Classical phase space anomaly]
Assume that \(\g_\partial\) acts symplectically on \((\cP_\Sigma,\Omega_\Sigma)\) and admits Hamiltonians \(H_\epsilon\).  The \emph{non-equivariance cocycle} is
\begin{equation}
  K_\Sigma(\epsilon_1,\epsilon_2)
  :=\{H_{\epsilon_1},H_{\epsilon_2}\}-H_{[\epsilon_1,\epsilon_2]}.
  \label{eq:Kdefinition}
\end{equation}
A nonzero class
\(
 [K_\Sigma]\in H^2_{\CE}(\g_\partial;\R)
\)
is called a \emph{classical phase space anomaly}.\\

\end{definition}

The cocycle \(K_\Sigma\) has the following elementary properties:

\begin{proposition}
\label{prop:classical-cocycle}
Under the preceding assumptions, \(K_\Sigma\) is constant on each connected component of \(\cP_\Sigma\), is a Chevalley--Eilenberg two-cocycle and changes by a CE coboundary under constant improvements \(H_\epsilon\mapsto H_\epsilon+\ell(\epsilon)\).
\end{proposition}

\begin{proof}
The Hamiltonian vector field of \(K_\Sigma(\epsilon_1,\epsilon_2)\) is
\[
  [v_{\epsilon_1},v_{\epsilon_2}]-v_{[\epsilon_1,\epsilon_2]}=0,
\]
so \(\del K_\Sigma=0\). The Jacobi identity for the Poisson bracket implies \(\dd_{\CE}K_\Sigma=0\).  Finally, a constant linear shift changes \(K_\Sigma\) by \(-\dd_{\CE}\ell\). Hence, the cohomology class is intrinsic.
\end{proof}

For an Abelian boundary algebra and trivial coefficients, every CE two-coboundary vanishes. Consequently, any nonzero antisymmetric bilinear form \(K_\Sigma\) is automatically cohomologically nontrivial. The statement remains dependent on the boundary conditions: if the allowed gauge parameters vanish at \(S\), the cocycle vanishes and no boundary symmetry remains.

\subsection{The Weil--CPS bicomplex}
We now introduce the Weil-CPS bicomplex and represent the classical anomaly in this bicomplex structure.\\
Let \(W(\g_\partial)\) be the Weil algebra of the boundary symmetry algebra, see \cref{app:weil} for detailed definition. For a basis \(e_a\) with structure constants \(f^a{}_{bc}\), it is generated by degree one elements \(\vartheta^a\) and degree two elements \(u^a\), with
\begin{equation}
  \dd_W\vartheta^a
  =u^a-\frac12 f^a{}_{bc}\vartheta^b\vartheta^c,
  \qquad
  \dd_Wu^a=-f^a{}_{bc}\vartheta^b u^c.
  \label{eq:WeilDifferential}
\end{equation}
We introduce the bicomplex
\begin{equation}
  \cW^{r,q}_\Sigma
  :=W^r(\g_\partial)\otimes\Omega^q(\cP_\Sigma),
  \qquad
  D_{W\!\cP}=\dd_W+(-1)^r\del.
  \label{eq:WeilCPS}
\end{equation}
The total contractions and Lie derivatives are
\[
  \mathbb I_a=\ii_a^W+\ii_{v_a},
  \qquad
  \mathbb L_a=[D_{W\!\cP},\mathbb I_a].
\]
The basic subcomplex is the common kernel of \(\mathbb I_a\) and \(\mathbb L_a\). Its cohomology is the Weil model of equivariant cohomology.

The equivalent Cartan description is more transparent for the present purpose.  If \(H:\g_\partial\to C^\infty(\cP_\Sigma)\) is a Hamiltonian assignment, define the Cartan element
\begin{equation}
  \Omega_C(\xi):=\Omega_\Sigma-H_\xi.
  \label{eq:CartanExtension}
\end{equation}
With \(\dd_C=\del-\ii_{v_\xi}\), the Hamiltonian equation \eqref{eq:HamLift} is exactly \(\dd_C\Omega_C=0\). The remaining requirement is equivariance:
\begin{equation}
  \Lie_{v_\epsilon}H_\eta-H_{[\epsilon,\eta]}=0.
  \label{eq:equivariance}
\end{equation}
Using \eqref{eq:HamConvention}, the left hand side equals \(K_\Sigma(\epsilon,\eta)\). Thus, the classical phase space anomaly admits the following Weil interpretation:

\begin{proposition}[Weil interpretation]
\label{prop:WeilInterpretation}
The classical anomaly class \([K_\Sigma]\) is the obstruction to lifting \(\Omega_\Sigma\) to a basic closed element of the Weil--CPS bicomplex. The Hamiltonian equation supplies Cartan closure, while the non-equivariance cocycle is precisely the defect of basicness.
\end{proposition}

This formulation makes the classical character of the construction manifest.  It uses only \(\Omega_\Sigma\), the classical symmetry action and the Weil model of equivariant cohomology.

\section{BRST descent and the classical shadow map}
\label{sec:shadow}
In this section we will give a basic review of the BRST anomaly and its descent theory. Then we will construct an antisymmetric multilinear polarization from ghost number expression, which we will call \emph{deghostification}. To bridge the BRST anomaly and our classical anomaly, we will focus on one special kind of BRST descent, called Hamiltonian admissible mixed descent, which we will define in the second part of this section. With this specific descent, we give and prove our main theorem that the classical anomaly will be the integration of our \emph{deghostification} of the Hamiltonian admissible mixed descent on a hypersurface corner. At the end of this section, we will give the descent-Hamiltonian identity which will give the relation between the first descendant and presymplectic current.
\subsection{Local BRST anomalies and their second descendants}

Let \(\Omega^{p,q,g}_{\mathrm{loc}}\) denote local forms with horizontal degree \(p\), field space degree \(q\) and ghost number \(g\). The three odd differentials obey
\begin{equation}
  \dH^2=\del^2=s^2=0,
  \quad
  \dH\del+\del\dH=0,
  \quad
  \dH s+s\dH=0,
  \quad
  \del s+s\del=0.
  \label{eq:tricomplex}
\end{equation}
A local consistent anomaly candidate is represented by
\begin{equation}
  a_d^1\in\Omega^{d,0,1}_{\mathrm{loc}},
  \qquad
  s a_d^1+\dH a_{d-1}^2=0.
  \label{eq:WZ}
\end{equation}
Equivalent representatives differ by
\begin{equation}
  a_d^1\sim a_d^1+s m_d^0+\dH n_{d-1}^1.
  \label{eq:BRSTequiv}
\end{equation}
The class belongs to \(H^{1,d}(s\mid\dH)\). It becomes a genuine quantum anomaly when it is realized with nonzero coefficient in the effective action or partition function.  The BRST complex itself remains a classical homological encoding of the gauge algebra.

For a ghost-number-\(r\) expression \(\alpha(C)\), let
\(
 \Pol_r(\alpha)(\epsilon_1,\ldots,\epsilon_r)
\)
denote its antisymmetric multilinear polarization, followed by replacement of ghosts by ordinary classical parameters. We refer to this antisymmetric polarization followed by ghost replacement as \emph{deghostification}.  \\

\begin{definition}[\emph{Deghostification} by antisymmetric polarization]
Let
\begin{equation}
\alpha(C)\in\Omega_{\mathrm{loc}}^{p,q,r}
\end{equation}
be a local expression of ghost number $r$, and suppose that its relevant
ghost-number-$r$ component is homogeneous of degree $r$ in the
degree-one ghosts, with no contribution from higher-stage ghosts. 
For ordinary classical symmetry parameters
\begin{equation}
\epsilon_1,\ldots,\epsilon_r\in\mathfrak g_{\partial},
\end{equation}
introduce $r$ auxiliary Grassmann-odd variables
\begin{equation}
\vartheta_1,\ldots,\vartheta_r,
\qquad
\vartheta_i\vartheta_j=-\vartheta_j\vartheta_i,
\qquad
\vartheta_i^2=0,
\end{equation}
and define
\begin{equation}
C_{\vartheta}
:=
\sum_{i=1}^{r}\vartheta_i\epsilon_i.
\end{equation}

The auxiliary variables $\vartheta_i$ belong to an external Grassmann
algebra. They anticommute with one another but commute with all local
horizontal and field-space forms. Their sole purpose is to antisymmetrize
the parameter labels.

The degree-$r$ deghostification, or antisymmetric polarization, of
$\alpha$ is defined by
\begin{equation}
\boxed{
\operatorname{Pol}_r(\alpha)
(\epsilon_1,\ldots,\epsilon_r)
:=
[\vartheta_1\cdots\vartheta_r]\,
\alpha(C_{\vartheta}),
}
\end{equation}
where
\(
[\vartheta_1\cdots\vartheta_r]
\)
denotes extraction of the coefficient of the ordered monomial
$\vartheta_1\cdots\vartheta_r$.

The same substitution is understood for every horizontal derivative of
the degree-one ghosts:
\begin{equation}
d_H^k C
\longmapsto
\sum_{i=1}^{r}
\vartheta_i\,d_H^k\epsilon_i,
\qquad
k\geq 1,
\end{equation}
with all products evaluated according to the total Koszul sign convention
of the local BRST complex.
\end{definition}

By construction, the polarization is alternating:
\begin{equation}
\operatorname{Pol}_r(\alpha)
(\epsilon_{\sigma(1)},\ldots,\epsilon_{\sigma(r)})
=
\operatorname{sgn}(\sigma)\,
\operatorname{Pol}_r(\alpha)
(\epsilon_1,\ldots,\epsilon_r),
\qquad
\sigma\in S_r.
\end{equation}
Hence
\begin{equation}
\operatorname{Pol}_r(\alpha)
\in
\operatorname{Hom}\!\left(
\wedge^r\mathfrak g_{\partial},
\Omega_{\mathrm{loc}}^{p,q,0}
\right).
\end{equation}

Equivalently, if the relevant ghost-number-$r$ component can be written
in terms of a local $r$-linear differential operator
$\mathcal B_r$ as
\begin{equation}
\alpha(C)
=
\mathcal B_r(C,\ldots,C),
\end{equation}
then
\begin{equation}
\boxed{
\operatorname{Pol}_r(\alpha)
(\epsilon_1,\ldots,\epsilon_r)
=
\sum_{\sigma\in S_r}
\operatorname{sgn}(\sigma)\,
\mathcal B_r
\bigl(
\epsilon_{\sigma(1)},\ldots,
\epsilon_{\sigma(r)}
\bigr),
}
\end{equation}

with the appropriate Koszul signs included whenever the arguments carry
nontrivial horizontal degree.

For reducible higher-form symmetries, the ghost number alone does not
determine the number of degree-one ghost insertions. In particular,
ghosts for ghosts and higher-stage ghosts encode reducibility rather than
ordinary $r$-cochains of the boundary symmetry algebra \cite{HenneauxKnaepen1999}. Accordingly,
$\operatorname{Pol}_r$ is applied, when necessary, after choosing within
the relevant relative BRST cohomology class a representative whose
component of interest is $r$-linear in degree-one ghosts. Such a
representative defines an ordinary Chevalley--Eilenberg $r$-cochain
through
\begin{equation*}
\operatorname{Pol}_r(\alpha)
:
\wedge^r\mathfrak g_{\partial}
\longrightarrow
\Omega_{\mathrm{loc}}^{p,q,0}.
\end{equation*}

In the remainder of this paper, we restrict attention to the ghost-number-two case; the higher-degree cases admit analogous generalizations.

\subsection{Hamiltonian admissible mixed descents}

The mixed anomalies of interest naturally involve a decomposition
\begin{equation}
  \g_\partial=\g_L\oplus\g_R,
  \qquad [\g_L,\g_R]=0.
  \label{eq:splitg}
\end{equation}
The full action of
$\g_\partial=\g_L\oplus\g_R$ on
$(\cP_\Sigma,\Omega_\Sigma)$ is symplectic and admits a general Hamiltonian
lift with
$$H=H_L\oplus H_R
$$
Since the Hamiltonian here is not a single lift of whether $\epsilon_L$ or $\epsilon_R$, it is hard to figure out the relation between the Hamiltonian and $\Omega_\Sigma$. However, a Bardeen type choice of local counterterm may place the ghost-number-one anomaly entirely in the right factor while the second descendant remains bilinear in left and right ghosts \cite{Bardeen1984}. This motivates the following notion of Hamiltonian admissibility, which singles out the descents for which the BRST data can be compared directly with the classical charge algebra.

\begin{definition}[Hamiltonian admissible mixed descent]
\label{def:admissible}
A descent pair \((a_d^1,a_{d-1}^2)\) is \emph{Hamiltonian admissible on \(\Sigma\)} if:
\begin{enumerate}[label=(\roman*)]
  \item \(a_d^1\) is linear in the right ghost \(C_R\) and contains no left ghost;
  \item for \(\eta\in\g_R\), the deghostified density
  \(
    A_R(\eta):=\Pol_1(a_d^1)(\eta)
  \)
  defines
  \[
    H_R(\eta):=\int_\Sigma A_R(\eta)
  \]
  satisfying \(\ii_{v_\eta}\Omega_\Sigma=-\del H_R(\eta)\);

  \item \(\Pol_2(a_{d-1}^2)(\epsilon,\eta)\) is a local \((d-1)\)-form for \(\epsilon\in\g_L\), \(\eta\in\g_R\).
\end{enumerate}
\end{definition}

The first condition is a representative choice, not a statement that the anomaly belongs to only one symmetry factor. In a mixed anomaly, local counterterms redistribute the first descendant between the two factors, while the joint class remains invariant.

\subsection{The shadow theorem}
For a Hamiltonian admissible mixed descent, the classical non-equivariance cocycle is obtained by \emph{deghostifying} the second BRST descendant and integrating it over the corner of a hypersurface.
\begin{theorem}[Classical shadow theorem]
\label{thm:shadow}
Let \((a_d^1,a_{d-1}^2)\) be a Hamiltonian admissible mixed descent on a hypersurface \(\Sigma\) with closed corner \(S=\partial\Sigma\). For \(\epsilon\in\g_L\) and \(\eta\in\g_R\), define
\begin{equation}
  K_{LR}(\epsilon,\eta)
  :=\Omega_\Sigma(v_\epsilon,v_\eta).
  \label{eq:KLR}
\end{equation}
Then
\begin{equation}
  \boxed{
  K_{LR}(\epsilon,\eta)
  =-\int_S \Pol_2(a_{d-1}^2)(\epsilon,\eta).}
  \label{eq:ShadowFormula}
\end{equation}
After antisymmetric extension to \(\g_L\oplus\g_R\), this is the classical charge cocycle. Its cohomology class depends only on the relative BRST class of the descent pair, up to a CE coboundary.
\end{theorem}

\begin{proof}
Polarizing the second descent equation \eqref{eq:WZ} in one left parameter and one right parameter gives
\begin{equation}
  \Lie_{v_\epsilon}A_R(\eta)
  +\dH\Pol_2(a_{d-1}^2)(\epsilon,\eta)=0.
  \label{eq:polarizedDescent}
\end{equation}
There is no second term \(\Lie_{v_\eta}A_L(\epsilon)\), because the chosen representative has no left-ghost component, and the cross bracket vanishes by \eqref{eq:splitg}.  Integrating over \(\Sigma\) yields
\begin{equation}
  \Lie_{v_\epsilon}H_R(\eta)
  =-\int_S\Pol_2(a_{d-1}^2)(\epsilon,\eta).
  \label{eq:integratedDescent}
\end{equation}
On the other hand, Hamiltonian admissibility and \eqref{eq:HamConvention} imply
\[
  \Lie_{v_\epsilon}H_R(\eta)
  =\del H_R(\eta)(v_\epsilon)
  =-\Omega_\Sigma(v_\eta,v_\epsilon)
  =\Omega_\Sigma(v_\epsilon,v_\eta).
\]
This proves \eqref{eq:ShadowFormula}.

Under \eqref{eq:BRSTequiv}, the second descendant changes by a polarized CE-exact term plus a horizontal exact term. Integration over closed \(S\) removes the horizontal exact term, while the remaining contribution is a CE coboundary, for detail, see \cref{app:independence}. Hence the induced class is independent of the chosen Bardeen representative.
\end{proof}

\begin{definition}[Classical shadow map]
On the subspace of Hamiltonian admissible transgressive anomaly classes, define
\begin{equation}
  \mathfrak S_{\Sigma,S}([a_d^1])
  :=\left[-\int_S\Pol_2(a_{d-1}^2)\right]
  \in H^2_{\CE}(\g_\partial;\R).
  \label{eq:shadowmap}
\end{equation}
\end{definition}

The classical shadow theorem says that \(\mathfrak S_{\Sigma,S}([a_d^1])\) is not an abstractly assigned cocycle: it is the non-equivariance class of the classical covariant phase space charges. This is the precise meaning of the word \emph{shadow} we defined in this paper.

\subsection{Local descent--Hamiltonian compatibility}

The preceding result is the integrated, charge algebra statement. It is complemented by a local identity that directly relates the first descendant to the presymplectic current. Let \(Q\) be the BRST evolutionary vector field, with field space Cartan relation
\begin{equation}
  s=[\ii_Q,\del]=\ii_Q\del-\del\ii_Q.
  \label{eq:CartanBRST}
\end{equation}

\begin{proposition}[Descent--Hamiltonian identity]
\label{prop:DHidentity}
Suppose
\[
  \del L=\mathcal E+\dH\theta,
  \qquad \omega=\del\theta,
  \qquad sL+\dH a_d^1=0.
\]
Define \(\mu=a_d^1+\ii_Q\theta\). Then
\begin{equation}
  \dH\bigl(\ii_Q\omega-\del\mu\bigr)=s\mathcal E.
  \label{eq:DHidentity}
\end{equation}
On shell, \(\ii_Q\omega-\del\mu\simeq\dH k\) locally.
\end{proposition}

\begin{proof}
Apply \(\del\) to \(sL+\dH a_d^1=0\), use \eqref{eq:tricomplex} and substitute \eqref{eq:firstvariation}.  One obtains
\(
 \dH(s\theta-\del a_d^1)=s\mathcal E
\).
Using \eqref{eq:CartanBRST},
\(
 s\theta=\ii_Q\omega-\del\ii_Q\theta
\), which gives \eqref{eq:DHidentity}.
\end{proof}

The first descendant therefore supplies a relative Hamiltonian density, while the second descendant supplies the charge algebra cocycle. The two statements are successive layers of the same descent.

\section{Mixed anomaly: the complete calculation}
\label{sec:example}
We now illustrate the shadow theorem in a five dimensional inflow model. The calculation exhibits the classical charge cocycle and the BRST second descendant as two realizations of the same Weil transgression class.
\subsection{Higher Weil data and the inflow transgression}

Let \(Y\) be five dimensional and let
\[
  B_E,B_M\in\Omega^2(Y),
  \qquad H_E=\dH B_E,
  \qquad H_M=\dH B_M.
\]
The Abelian higher Weil algebra \cite{SatiSchreiberStasheff2008} has degree-two generators \(b_E,b_M\) and degree-three curvatures \(h_E,h_M\), with
\begin{equation}
  \dd_W b_I=h_I,
  \qquad \dd_W h_I=0.
  \label{eq:higherWeil}
\end{equation}
The mixed invariant polynomial and a transgression element are
\begin{equation}
  I_6=\kappa h_Eh_M,
  \qquad
  \operatorname{cs}_5=-\kappa h_Eb_M,
  \qquad
  \dd_W\operatorname{cs}_5=I_6.
  \label{eq:universalMixed}
\end{equation}
The classical realization is
\begin{equation}
  L_5=-\kappa H_E\wedge B_M,
  \qquad
  \dH L_5=\kappa H_E\wedge H_M.
  \label{eq:L5}
\end{equation}
For the Maxwell normalization one may take \(\kappa=(2\pi)^{-2}\).

\subsection{Classical covariant phase space}

Using \(\del H_E=-\dH\del B_E\) and total Koszul grading, the variation of \eqref{eq:L5} is
\begin{align}
  \del L_5
  &=\kappa\bigl(\del B_E\wedge H_M+H_E\wedge\del B_M\bigr)
    +\dH\theta_4,\label{eq:variationL5}\\
  \theta_4&=\kappa\,\del B_E\wedge B_M,\label{eq:theta4}\\
  \omega_4&=\del\theta_4
  =-\kappa\,\del B_E\wedge\del B_M.
  \label{eq:omega4}
\end{align}
Let \(\Sigma\subset Y\) be a four dimensional hypersurface and \(S=\partial\Sigma\). The classical presymplectic form is
\begin{equation}
  \Omega_\Sigma
  =-\kappa\int_\Sigma\del B_E\wedge\del B_M.
  \label{eq:OmegaExample}
\end{equation}

Take ordinary classical one-form gauge parameters \(\Lambda_E,\Lambda_M\) and vector fields
\begin{equation}
  v_E(\Lambda_E)B_E=\dH\Lambda_E,
  \quad v_E(\Lambda_E)B_M=0,
  \qquad
  v_M(\Lambda_M)B_M=\dH\Lambda_M,
  \quad v_M(\Lambda_M)B_E=0.
  \label{eq:gaugevectors}
\end{equation}
With the convention \eqref{eq:HamConvention}, Hamiltonians are
\begin{align}
  H_E[\Lambda_E]
  &=\kappa\int_\Sigma\dH\Lambda_E\wedge B_M,
  \label{eq:HEcharge}\\
  H_M[\Lambda_M]
  &=-\kappa\int_\Sigma B_E\wedge\dH\Lambda_M.
  \label{eq:HMcharge}
\end{align}
Indeed,
\(
 \ii_{v_E}\Omega_\Sigma=-\del H_E
\)
and
\(
 \ii_{v_M}\Omega_\Sigma=-\del H_M
\).
On the topological equations of motion \(H_E\simeq H_M\simeq0\), these become corner charges:
\begin{equation}
  H_E\simeq\kappa\int_S\Lambda_E\wedge B_M,
  \qquad
  H_M\simeq-\kappa\int_S B_E\wedge\Lambda_M,
  \label{eq:cornercharges}
\end{equation}
up to the orientation conventions of Stokes' theorem.

\subsection{The classical mixed charge anomaly}

The two gauge algebras commute, but their Hamiltonians do not:
\begin{align}
  K_{EM}(\Lambda_E,\Lambda_M)
  &:=\{H_E[\Lambda_E],H_M[\Lambda_M]\}
  \nonumber\\
  &=\Omega_\Sigma(v_E(\Lambda_E),v_M(\Lambda_M))
  \nonumber\\
  &=-\kappa\int_\Sigma\dH\Lambda_E\wedge\dH\Lambda_M
  \nonumber\\
  &=-\kappa\int_S\Lambda_E\wedge\dH\Lambda_M.
  \label{eq:KEM}
\end{align}
For general pairs \(\lambda=(\Lambda_E,\Lambda_M)\) and \(\lambda'=(\Lambda'_E,\Lambda'_M)\), the antisymmetric extension is
\begin{equation}
  K_{\mathrm{CPS}}(\lambda,\lambda')
  =-\kappa\int_\Sigma
  \bigl(
  \dH\Lambda_E\wedge\dH\Lambda'_M
  -\dH\Lambda'_E\wedge\dH\Lambda_M
  \bigr).
  \label{eq:fullK}
\end{equation}
It is field independent, reducibility invariant under \(\Lambda_I\mapsto\Lambda_I+\dH\lambda_I\), and a CE two-cocycle. For the Abelian boundary algebra, a nonzero \eqref{eq:fullK} is a nontrivial central extension class. This is the classical phase space anomaly.

With related electromagnetic boundary central extensions \cite{FreidelPranzetti2018,HofmanIqbal2018}, we will identify it with a higher BRST second descendant of the mixed one-form anomaly below.

\subsection{Original and improved higher BRST descents}

Introduce two reducible BRST towers
\begin{equation}
  sB_I=\dH C_I,
  \qquad
  sC_I=\dH c_I,
  \qquad
  sc_I=0,
  \qquad I=E,M.
  \label{eq:BRSTtowers}
\end{equation}
The higher Russian formula \cite{JiaWangZhang2026I} is
\begin{equation}
  (\dH+s)(B_I-C_I+c_I)=H_I.
  \label{eq:Russian}
\end{equation}
The direct realization of \(\operatorname{cs}_5\) gives the familiar descent
\begin{equation}
  L_5=-\kappa H_EB_M,
  \qquad
  a_4^1=\kappa H_EC_M,
  \qquad
  a_3^2=-\kappa H_Ec_M.
  \label{eq:originaldescent}
\end{equation}
This is a valid representative, but \(a_3^2\) contains the ghost for ghost and is not adapted to an ordinary two-charge cocycle. We now perform an explicit cohomologically equivalent improvement.

Set
\begin{equation}
  m_3^1=\kappa B_E\wedge C_M,
  \qquad
  n_2^2=\kappa\bigl(C_E\wedge C_M+B_E\wedge c_M\bigr).
  \label{eq:mnimprovement}
\end{equation}
A direct computation gives
\begin{equation}
  a_3^2
  =\kappa C_E\wedge\dH C_M
   +s m_3^1-\dH n_2^2.
  \label{eq:a3equivalence}
\end{equation}
Therefore the equivalent descent representative may be chosen as
\begin{align}
  \widetilde a_4^1
  &:=a_4^1-\dH m_3^1
  =-\kappa B_E\wedge\dH C_M,
  \label{eq:atilde4}\\
  \widetilde a_3^2
  &:=a_3^2-s m_3^1+\dH n_2^2
  =\kappa C_E\wedge\dH C_M,
  \label{eq:atilde3}\\
  \widetilde a_2^3
  &:=-\kappa c_E\wedge\dH C_M.
  \label{eq:atilde2}
\end{align}
They obey
\begin{equation}
  sL_5+\dH\widetilde a_4^1=0,
  \qquad
  s\widetilde a_4^1+\dH\widetilde a_3^2=0,
  \qquad
  s\widetilde a_3^2+\dH\widetilde a_2^3=0.
  \label{eq:improveddescent}
\end{equation}
The first descendant is a Bardeen representative that assigns the anomaly to the magnetic factor, while the second descendant is manifestly bilinear in the electric and magnetic degree-one ghosts \cite{Bardeen1984}. For detailed calculation, see \cref{app:improved}

\subsection{Exact equality between the shadow and the second descendant}

Deghostification of the improved first descendant gives
\begin{equation}
  \degh\bigl(\widetilde a_4^1\bigr)(\Lambda_M)
  =-\kappa B_E\wedge\dH\Lambda_M,
  \label{eq:deghfirst}
\end{equation}
whose integral is precisely the magnetic Hamiltonian \eqref{eq:HMcharge}.  Thus the descent is Hamiltonian admissible.

The polarized second descendant is
\begin{equation}
  \Pol_2\bigl(\widetilde a_3^2\bigr)(\Lambda_E,\Lambda_M)
  =\kappa\Lambda_E\wedge\dH\Lambda_M.
  \label{eq:deghsecond}
\end{equation}
Applying \Cref{thm:shadow} yields
\begin{equation}
  \boxed{
  K_{EM}(\Lambda_E,\Lambda_M)
  =-\int_S
  \Pol_2\bigl(\widetilde a_3^2\bigr)(\Lambda_E,\Lambda_M)
  =-\kappa\int_S\Lambda_E\wedge\dH\Lambda_M.}
  \label{eq:MainEqualityExample}
\end{equation}
Equation (4.25) gives the desired identification. The classical charge anomaly and the ghost-number-two BRST descendant are not merely analogous expressions; they arise from the same Weil transgression class under two distinct realizations:
\begin{equation}
\begin{CD}
  \kappa h_Eh_M @>{\text{higher BRST realization}}>>
  \widetilde a_3^2=\kappa C_E\dH C_M \\
  @V{\text{classical CPS realization}}VV
  @VV{\text{deghostify and integrate on }S}V \\
  K_{EM} @= -\displaystyle\int_S\degh(\widetilde a_3^2).
\end{CD}
\label{eq:commutingdiagram}
\end{equation}

\subsection{Relation to the relative Hamiltonian corner term}

For the original first descendant \(a_4^1=\kappa H_EC_M\), \Cref{prop:DHidentity} gives
\begin{equation}
  \ii_{Q_M}\omega_4-\del a_4^1
  =
  \dH\!\left(\kappa\,\del B_E\wedge C_M\right)
  +\kappa H_E\wedge\del C_M
  \simeq
  \dH k_3^{1,1},
  \qquad
  k_3^{1,1}
  =\kappa\,\del B_E\wedge C_M.
\end{equation}
Deghostifying \(C_M\mapsto\Lambda_M\) and contracting once more with the electric classical vector field gives
\begin{equation}
  \ii_{v_E}k_{\Lambda_M}
  =\kappa\,\dH\Lambda_E\wedge\Lambda_M.
  \label{eq:secondcontraction}
\end{equation}
On closed \(S\), integration by parts identifies this with \(\kappa\Lambda_E\wedge\dH\Lambda_M\).  Consequently,
\begin{equation}
  K_{EM}
  =-\int_S\ii_{v_E}k_{\Lambda_M}.
  \label{eq:cornerToK}
\end{equation}
The local relative Hamiltonian defect is therefore the first phase space descendant, while the classical charge cocycle is obtained by a second contraction. This matches the passage from the ghost-number-one to the ghost-number-two BRST descendant.

\section{Conclusion and Discussion}
\label{sec:discussion}

We have introduced a Weil covariant phase space bicomplex in which a classical anomaly is the obstruction to a basic equivariant extension of the covariant symplectic form. The obstruction is the non-equivariance cocycle of the Hamiltonian charge algebra. We then proved a shadow theorem: for a Hamiltonian admissible mixed BRST descent, the classical charge cocycle is the Cauchy-transgressed, deghostified ghost-number-two descendant. The construction is invariant under the usual changes of descent representative up to CE coboundary.

The construction involves three conceptually distinct objects: the classical phase space cocycle \([K_\Sigma]\), defined solely from \((\cP_\Sigma,\Omega_\Sigma)\) and the classical symmetry action. The local BRST class \([a_d^1]\), a classical cohomological encoding of possible consistent anomalies. The quantum anomaly, realized when a quantum effective action or partition function transforms by the nonzero class \([a_d^1]\). The shadow theorem identifies  the classical phase space cocycle with a descendant of the local BRST class under explicit Hamiltonian admissibility assumptions.

For the mixed electric--magnetic one-form anomaly, the classical covariant phase space of the five dimensional transgression carries
\[
  K_{EM}(\Lambda_E,\Lambda_M)
  =-\kappa\int_S\Lambda_E\wedge\dH\Lambda_M,
\]
while the improved BRST descent has
\[
  \widetilde a_3^2=\kappa C_E\wedge\dH C_M.
\]
Their equality under deghostification and corner integration proves, in this model, that the classical phase space anomaly is a genuine shadow of the BRST anomaly class.

The result should be viewed neither as a redefinition of a quantum anomaly nor as a universal quantization theorem. Its content is a precise bridge between two cohomological realizations of the same transgression data: failure of basic Hamiltonian equivariance in the classical Weil--CPS complex, and the second descent of a consistent BRST anomaly.

\subsection*{The Weil bicomplex and the inflow example}

The same universal invariant polynomial is mapped into two different complexes.  The classical realization lands in the Weil--CPS bicomplex and tests whether a symplectic form admits a basic equivariant extension. The BRST realization maps Weil generators to fields, ghosts and ghosts for ghosts and produces the Stora--Zumino  descent \cite{Zumino1984}. The shadow theorem compares the degree-two symmetry components after Cauchy transgression.  Schematically,
\begin{equation}
\begin{CD}
  W(\mathfrak L) @>{\rho_{\mathrm{BRST}}}>>
  \bigl(\Omega^{\bullet,\bullet}_{\mathrm{loc}},\dH+s\bigr) \\
  @V{\rho_{\mathrm{CPS}}}VV @VV{\Pol_2\text{ and }\int_S}V \\
  \bigl(W(\g_\partial)\otimes\Omega(\cP_\Sigma)\bigr)_{\mathrm{basic}?
  } @>>{\text{basicness defect}}> H^2_{\CE}(\g_\partial;\R).
\end{CD}
\label{eq:masterdiagram}
\end{equation}
The question mark emphasizes that an anomaly is precisely the failure of a basic lift. This is the sense in which the classical and BRST structures are two descendants of the same Weil data.\\

Our work isolates the non-equivariance class as the basicness obstruction in a Weil--CPS bicomplex and identifies it, at the level of representatives and cohomology, with the ghost-number-two descendant of a higher-form anomaly polynomial. This comparison is also analogous to the familiar relation between consistent anomalies and Schwinger terms: the ghost-number-one object governs the anomalous Ward identity, whereas the ghost-number-two descendant governs the current or charge algebra \cite{Faddeev1984}. Our result is classical because the latter algebra is a Poisson algebra on \(\cP_\Sigma\), not an operator algebra on a Hilbert space.\\

Our example uses the classical covariant phase space of a five dimensional inflow theory. The four dimensional quantum anomaly is encoded by a five dimensional transgression, while the classical charge cocycle lives on the corner of a four dimensional hypersurface in the bulk. Thus the result is a phase space refinement of anomaly inflow \cite{Callan1985}:
\[
  \text{bulk transgression}
  \longrightarrow
  \begin{cases}
    \text{boundary BRST anomaly descent},\\
    \text{bulk classical CPS charge extension}.
  \end{cases}
\]

\subsection*{The structure analogy and assumptions}

The guiding analogy is structural. Gromov non-squeezing is a purely classical theorem, yet symplectic capacities provide a geometric skeleton for uncertainty relations once a quantum scale is introduced \cite{Gromov1985,deGosson2002,deGossonLuef2006}. Here the classical non-equivariance class provides a geometric skeleton for a BRST anomaly descent. In both cases the classical geometry does not reproduce the quantum phenomenon in full; it constrains and organizes the structures that quantization can realize.\\

The theorem applies to Hamiltonian admissible transgressive anomaly classes. It is natural to ask whether more general quantum anomaly classes admit analogous classical phase space shadows. The present analysis is local and perturbative, so for global gerbe data, differential cohomology and torsion anomalies may require a Deligne or bordism refinement. In the end, if the corner itself has a boundary, further descendants and higher-codimension data may enter \cite{JiaWangZhang2026II}.

\section*{\centering Acknowledgment}
This work was supported by the National Natural Science Foundation of China with the Grant No. 12375049 and No. 12405059, Key Program of the Natural Science Foundation of Jiangxi
Province under Grant No. 20232ACB201008, and the
Ganpo High-Level Innovative Talent Program.

\section*{Appendix}

\appendix

\section{Sign and grading conventions}
\label{app:signs}

We use the tridegree \((p,q,g)\) of horizontal degree, field space degree and ghost number, with total parity
\[
  |\alpha|=p+q+g\pmod 2.
\]
The differentials \(\dH\), \(\del\) and \(s\) are odd and pairwise anticommute. The BRST contraction \(\ii_Q\) has tridegree \((0,-1,+1)\) and total degree zero, so it is an even derivation of the total graded product. The field space Cartan relation is \(s=[\ii_Q,\del]\).

On integrated phase space we use \(\ii_{X_F}\Omega=-\del F\) and \(\{F,G\}=\Omega(X_F,X_G)\). Reversing this convention reverses several displayed overall signs but does not change the cohomology classes or the shadow statement.

\section{Weil, Cartan and BRST models}
\label{app:weil}

The Weil algebra \(W(\g)\) is acyclic, but its basic subcomplex in \(W(\g)\otimes\Omega(M)\) computes equivariant cohomology. The Mathai--Quillen/Kalkman isomorphism identifies the basic Weil model with the Cartan model \cite{Mathai1986}
\[
  \bigl(S(\g^*[2])\otimes\Omega(M)\bigr)^\g,
  \qquad
  \dd_C=\dd-u^a\ii_{v_a}.
\]
For a Hamiltonian action, \(\Omega-H\) is Cartan closed. It belongs to the invariant Cartan complex if and only if the moment map is equivariant. The defect is the cocycle \(K\).

BRST ghosts are a Chevalley--Eilenberg realization of the degree-one Weil generators. In a reducible two-form theory, the higher Weil generator of degree two is realized by the total field
\[
  \widehat B=B-C+c,
  \qquad
  (\dH+s)\widehat B=H.
\]
The curvature generator maps to \(H\). Expanding a universal transgression in ghost number gives the full higher descent.

\section{Representative independence of the shadow class}
\label{app:independence}

Suppose
\[
  a_d^1\mapsto a_d^1+s m_d^0+\dH n_{d-1}^1.
\]
One can choose the next descendant so that
\[
  a_{d-1}^2\mapsto a_{d-1}^2+s n_{d-1}^1+\dH r_{d-2}^2.
\]
After polarization and integration over closed \(S\), the \(\dH r\) term vanishes and the \(s n\) term becomes a CE coboundary. Therefore
\[
  \left[-\int_S\Pol_2(a_{d-1}^2)\right]
\]
is a well defined CE class. This is the cohomological version of the familiar statement that Bardeen counterterms move an anomaly between symmetry factors without removing the mixed class.

\section{Detailed verification of the improved descent}
\label{app:improved}

Using \eqref{eq:BRSTtowers},
\begin{align*}
  s m_3^1
  &=\kappa\bigl(\dH C_E\wedge C_M+B_E\wedge\dH c_M\bigr),\\
  \dH n_2^2
  &=\kappa\bigl(
    \dH C_E\wedge C_M
    +C_E\wedge\dH C_M
    +H_E\wedge c_M
    +B_E\wedge\dH c_M
  \bigr).
\end{align*}
Hence
\[
  s m_3^1-\dH n_2^2
  =-\kappa C_E\wedge\dH C_M-\kappa H_E\wedge c_M,
\]
which proves \eqref{eq:a3equivalence}.  Furthermore,
\begin{align*}
  s\widetilde a_4^1
  &=-\kappa\dH C_E\wedge\dH C_M,
  &\dH\widetilde a_3^2
  &=+\kappa\dH C_E\wedge\dH C_M,\\
  s\widetilde a_3^2
  &=+\kappa\dH c_E\wedge\dH C_M,
  &\dH\widetilde a_2^3
  &=-\kappa\dH c_E\wedge\dH C_M.
\end{align*}
This verifies the full improved descent.

\section{A concise dictionary}
\label{app:dictionary}

\begin{longtable}{>{\raggedright\arraybackslash}p{0.25\textwidth} >{\raggedright\arraybackslash}p{0.31\textwidth} >{\raggedright\arraybackslash}p{0.31\textwidth}}
\toprule
Structure & Classical CPS realization & BRST/quantum realization \\
\midrule
\endfirsthead
\toprule
Structure & Classical CPS realization & BRST/quantum realization \\
\midrule
\endhead
Invariant polynomial & Universal source of the symplectic transgression & Anomaly polynomial \(I_{d+2}\) \\
First descendant & Hamiltonian density or relative Hamiltonian & Consistent anomaly representative \(a_d^1\) \\
Second descendant & Charge non-equivariance cocycle \(K\) & Schwinger term type descendant \(a_{d-1}^2\) \\
Weil condition & Basic equivariant extension of \(\Omega\) & BRST/CE closure and Wess--Zumino consistency \\
Nontriviality & Central or higher extension of classical charges & Obstruction to quantum gauging when the class is realized \\
\bottomrule
\end{longtable}

\end{document}